\documentclass[12pt,draftclsnofoot,onecolumn,doublespaced]{IEEEtran}

\ifCLASSINFOpdf
\else
\fi
\usepackage{cite}
\usepackage{url}
\usepackage{graphicx}
\usepackage{amsmath}
\usepackage{amsfonts}
\usepackage{array}
\usepackage[tight]{subfigure}
\usepackage{amscd}
\usepackage{latexsym}
\usepackage{multirow}
\usepackage{graphics}
\usepackage{cite,algorithm,algorithmic}
\usepackage{makeidx}
\usepackage{color}
\usepackage{amsthm}
\usepackage{amssymb}

\date{}

\newcommand{\eqdef}{\stackrel{\triangle}{=}}
\newcommand{\be}{\begin{equation}}
\newcommand{\ba}{\begin{array}{lcl}}
\newcommand{\bs}{\begin{small}}
\newcommand{\es}{\end{small}}
\newcommand{\ea}{\end{array}}
\newcommand{\ee}{\end{equation}}
\newcommand{\bsp}{\mbox{ }}
\newcommand{\noi}{\noindent}

\newtheorem{lemma}{Lemma}

\newtheorem{theorem}{Theorem}

\newcommand{\revbegin}{\begin{color}{black}}
\newcommand{\revend}{\end{color}}

\begin{titlepage}
\title{
\vspace{-3cm}
\hfill {\em{\small Accepted to IEEE Trans. on Signal Processing}}\\
Feedback Allocation For OFDMA Systems With Slow Frequency-domain Scheduling\thanks{*Corresponding author. A preliminary version with a subset of the results was presented at the {\em Allerton Conference on Communication, Control and Computing}, Monticello, IL, Sep. 2009.}
}
{\small
{\author{Harish Ganapathy$^*$,~Siddhartha Banerjee,~Nedialko B. Dimitrov~and Constantine Caramanis%
\thanks{H. Ganapathy, S. Banerjee and C. Caramanis are with the Department of Electrical and Computer Engineering, The University of Texas, Austin, TX 78712. N. B. Dimitrov is with the Operations Research Department, Naval Postgraduate School, Monterey, CA 93943.}\\
\texttt{E-mail: \{harishg,siddhartha\}@utexas.edu,ned@nps.edu,\\ caramanis@mail.utexas.edu}}
}
}

\end{titlepage}

\begin{document}
\maketitle
\thispagestyle{empty}
\begin{abstract}
We study the problem of allocating limited feedback resources across multiple users in an orthogonal-frequency-division-multiple-access downlink system with slow frequency-domain scheduling. Many flavors of slow frequency-domain scheduling (e.g., persistent scheduling, semi-persistent scheduling), that adapt user-sub-band assignments on a slower time-scale, are being considered in standards such as 3GPP Long-Term Evolution. In this paper, we develop a feedback allocation algorithm that operates in conjunction with any arbitrary slow frequency-domain scheduler with the goal of improving the throughput of the system. Given a user-sub-band assignment chosen by the scheduler, the feedback allocation algorithm involves solving a weighted sum-rate maximization at each (slow) scheduling instant. We first develop an optimal dynamic-programming-based algorithm to solve the feedback allocation problem with pseudo-polynomial complexity in the number of users and in the total feedback bit budget. We then propose two approximation algorithms with complexity further reduced, for scenarios where the problem exhibits additional structure.
\end{abstract}

\begin{keywords}
Limited feedback, multi-user feedback allocation, throughput-optimal, uplink feedback, Random Vector Quantization, sub-modular functions, convex relaxations
\end{keywords}

\section{Introduction}\label{intro}
Orthogonal-frequency-division-multiple-access (OFDMA) is the multiple-access technology of choice for many current and future wireless standards such as 3GPP Long-Term Evolution (LTE), IEEE 802.16e (WiMAX) and Long-Term Evolution Advanced (LTE-A). With the singular goal of achieving higher throughputs in order to keep pace with the ever-growing suite of data-hungry applications, OFDMA systems typically operate in conjunction with a fast frequency-domain scheduler that allows for aggressive adaptation to the fading conditions of the channel. Here, user-sub-band assignments are typically made once every 1ms, 2ms or 5ms depending upon the standard under consideration. In the quest for higher data rates, the overhead incurred in \textit{enabling} such fast frequency-domain scheduling is often ignored.

Primarily, there are two types of overhead that facilitate user scheduling in an OFDMA downlink system. These are: the overhead incurred in \textit{(i)} communicating user-sub-band assignments and in \textit{(ii)} collecting channel state information (CSI) from all users commonly referred to as the process of feedback. To address the first issue, recently, there has been an increasing interest in ``slow'' frequency-domain scheduling \cite{intel,nokia,jin, ltebook1,win} instead of its faster counterpart for applications where the overhead demands of the latter do not justify its use. For example, LTE adopts \textit{(semi-)persistent} scheduling for voice-over-IP applications that typically do not have high throughput demands  \cite{intel,nokia,jin, ltebook1}. Here, user-sub-band assignments are decided on a slower time-scale while link adaptation (on the fast time-scale) specifically in the semi-persistent approach, is achieved through Hybrid Automatic Repeat Request (HARQ) re-transmissions. Li et al. \cite{win} show that slow OFDMA scheduling can achieve throughputs close to the ideal case in many real-world scenarios.

Moving on to the implications of \textit{(ii)}, we borrow an example of a typical LTE setting recently provided in Ouyang et al. \cite{ying}:  In LTE, the smallest unit of bandwidth that can be assigned to a user for data transmission is called a resource block, which is essentially a group of OFDM sub-carriers. If we consider a $10$MHz LTE system with $L=50$ resource blocks shared by $K = 50$ users equipped with standard $4$-bit codebooks (modulation/coding tables) at the mobiles, we have a total feedback bandwidth of $4KL = 4\times 50\times 50 = 10$kb per sub-frame \cite{3gpp}. Given a typical uplink data rate of $48$kb per sub-frame, this consumes $20\%$ of the uplink capacity, clearly making feedback bandwidth consumption an important bottleneck. This observation, amongst others, has motivated the development of limited feedback techniques \cite{love_tut,heath,jafarkhani1,jafarkhani2,mehta,agarwal,poor,chen,gesbert}. In general, adapting the size of the codebook (e.g., CSI table at the mobiles) \cite{heath,jafarkhani1,jafarkhani2} and sub-carrier grouping \cite{mehta,agarwal,poor,chen,gesbert}, subject to a constraint on the total available feedback bandwidth, are two of the most popular multi-user limited feedback approaches in the literature. In the former, the size of the codebook on each OFDMA sub-band, and potentially the codebooks themselves, are periodically chosen based on the ``state'' of the network. In the latter, feedback reduction is achieved through a grouping technique where one CSI report is generated for a group of OFDMA sub-bands.

In this paper, we propose a feedback allocation policy that operates in conjunction with a slow frequency-domain scheduler assumed given. In particular, given a scheduling assignment on a slower time-scale, i.e., once every $T$ time slots, the feedback allocation policy decides user codebook sizes again on the same time-scale. Thus, in the context of past literature, we focus on the former limited feedback approach of choosing dynamic codebook sizes as a function of the network state (e.g., channels, queues, etc.), a process that we call \textit{feedback allocation} henceforth\footnote{One can in general consider a more comprehensive feedback allocation policy that includes both codebook-size adaptation \textit{and} sub-carrier grouping. However, such policies are beyond the scope of this paper and a subject for future study.}, to address the second type of overhead. The difference between our approach and past work on dynamic codebook selection is that our algorithm adapts to queue sizes and hence user applications, in addition to the channel state thereby generalizing earlier methods.

The main contributions of this paper are the following:
\begin{enumerate}
\item[1)] We propose a throughput-optimal feedback allocation policy that overlays a given slow scheduler. The proposed policy takes the form of a weighted sum-rate maximization problem that needs to be solved once every $T$ time slots. Throughput-optimality is with respect to the space of all possible feedback allocation policies while \emph{fixing} the particular data scheduler of interest.
\end{enumerate}
Efficient algorithmic implementability of these policies is a critical design requirement, and this is the focus of our remaining contributions. Our focus is aligned with several papers over the last decade, which study the algorithmic aspects of queue-based scheduling for specific network structures and resource allocation problems (see, e.g., \cite{shroff,srikant,huang,tan1,tan2} and references therein). \revbegin Needless to say, the difficulty in solving the weighted sum-rate maximization problem in each slot is intricately tied to the resource being optimized. Recently, significant strides were made by Tan et al. \cite{tan1,tan2} in solving the joint queue-based scheduling and power control problem that has attracted much interest over the years (see \cite{huang} and references therein). Here, the possible transmission rates in each slot come from a continuous region induced by all possible power allocations. The authors \cite{tan1,tan2} solve non-convex power-control problems (e.g., weighted sum-rate) accurately and efficiently by using solutions to related convex problems (e.g., max-min rate) in an intelligent manner. Optimality is established under many channel settings. While the lack of convexity is due to interfering users in \cite{tan1,tan2}, in the model by Huang et al. \cite{huang}, self-noise due to channel estimation error forms the cause. The authors \cite{huang} nevertheless propose both optimal and sub-optimal approaches with varying degrees of complexity. In contrast to power allocation, in our case, the region of possible rates in each slot is discrete and is induced by all possible splits of the total feedback budget. This allows us to leverage many powerful  tools from the area of combinatorial optimization. With the exception of the work of Ouyang and Ying \cite{ying}, the problem of feedback allocations has not been considered in the past, to the best of our knowledge.
\revend

\begin{enumerate}
\item[2)] We develop a dynamic programming algorithm that solves the weighted sum-rate feedback allocation problem with pseudo-polynomial complexity in the number of users and the total feedback bit budget. This approach is exact and requires no assumptions on the structure of the weighted sum-rate function.\label{point1}

\item[3)] We show that in many practical wireless systems, the weighted sum-rate is non-decreasing and sub-modular in the feedback budget. Using this observation, we leverage sub-modular optimization results from combinatorial optimization (e.g. \cite{goundan,calinescu,vondrak}) and propose a reduced-complexity algorithm with an approximation guarantee of $(1-\frac{1}{e})$.

\item[4)] Multiple-input-single-output (MISO) beamforming is being considered as a potential transmission mode in the Long Term Evolution standard \cite{ltebook1}. For such systems, we show that when the popular Random Vector Quantization codebook \cite{raghavan,love_RVQ1,love_RVQ2} is used, we are able to reduce the complexity even further and provide an approximation guarantee of $\frac{1}{2}$.\label{point2}
\end{enumerate}

The rest of this paper is organized as follows. In Section \ref{sec:sysmodel}, we introduce the system model for feedback allocation and slow data scheduling. In Section \ref{sec:tputopt}, we discuss the notion of throughput-optimality in queueing networks and introduce a throughput-optimal joint feedback allocation and slow data scheduling policy. In Section \ref{sec:optsol}, we solve the optimal online feedback optimization problem for both objectives while in Section \ref{sec:redcomp}, we investigate methods of reducing the complexity of the optimal online optimization problem by exploiting more structure in the objective function. Simulation results are presented in Section~\ref{sec:simulations}. Concluding remarks are made in Section~\ref{sec:conclusion}.

{\em Notation}: $x_{ij}$ denotes element $(i,j)$ of matrix ${\bf X}$ while $x_{i}$ denotes element $i$ of vector ${\bf x}$. Given matrices ${\bf X},{\bf Y}\in{\mathbb R}^{p\times q}$, $X\leq Y$ means $x_{ij}\leq y_{ij},\forall i=1,\ldots,p,\bsp j=1,\ldots,q$. ${\mathbb R}_{+}$, ${\mathbb N}_{0}$ and ${\mathbb N}$ represent the non-negative real numbers, non-negative integers and positive integers respectively.

\section{System model}
\label{sec:sysmodel}
We consider the downlink of a frequency-division-duplex OFDMA system with $L$ sub-carriers/sub-bands and $K$ users that operates in slotted-time. The network model is described below:

\textbf{Channel State}: The true supportable rate for user $i$ on sub-band $j$ at time $t$ is given by $\mu_{ij}[t]$. We assume that $\mu_{ij}[t]$ is ergodic and comes from a finite but potentially large set ${\mathcal M}$. We assume that the mobile has perfect knowledge of the channel state $\{\mu_{ij}[t]\}_{j=1}^L$ in every time slot. The cumulative distribution function for rate $\mu_{ij}[t]$ is given by $\mbox{Pr}\left(\mu_{ij}[t]=m\right)=\pi_{mi}\left(\alpha_i\left[t\right]\right),\bsp m\in {\mathcal M}$, where $\alpha_i\left[t\right]$ denotes a large-scale fading gain that is dependent on user position and comes from a finite set $\alpha_i[t]\in\Omega$. Users change positions once every $T$ slots where $T\in\{1,2,3,\ldots\}$ denotes the large-scale fading coherence time. For ease of notation, we introduce a counter $\bar{t}=\lfloor\frac{t}{T}\rfloor T$ to keep track of the slower large-scale fading time-scale, i.e., $\pi_{mi}\left(\alpha_i[t]\right)=\pi_{mi}\left(\alpha_i[\bar{t}]\right),\bsp \forall t$. For convenience, we also set $\pi_{mi}[\bar{t}]=\pi_{mi}\left(\alpha_i[\bar{t}]\right)$ making implicit the dependence on $t$ and $T$. Note that the large-scale coefficient is typically only distance-dependent and independent of frequency allowing us to omit the index $j$ when representing it. We assume that the base station has perfect knowledge of $\{\alpha_{i}[\bar{t}]\}$ and all distribution information $\{\pi_{mi}[\bar{t}]\}$. \textit{Most importantly, $\bar{t}$ represents the time-scale at which feedback optimizations and scheduling assignments are decided.}

\textbf{Traffic model}: Each user $k\in\{1,2,\ldots,K\}$, has a queue of untransmitted packets with queue-length $q_k[\bar{t}]$ that is maintained at the base station with associated arrival rate $\lambda_k$.

\textbf{Feedback model}: The base station allocates a feedback budget of $b_k[\bar{t}]$ bits to user $k$ such that $\sum_{k=1}^K b_k[\bar{t}]\leq B$ where $B$ represents the total limited feedback budget of the system. Let the sub-carriers in our OFDMA system be indexed by ${\mathcal S}=\{1,2,\ldots,L\}$. Assume that user $k$ is allocated sub-bands ${\mathcal N}_k[\bar{t}]\subseteq {\mathcal S}$ by the slow scheduling algorithm. Given a budget of $b_k[\bar{t}]$ bits by the base station and an assignment ${\mathcal N}_k[\bar{t}]$ of size $|{\mathcal N}_k[\bar{t}]| = n_k$, the user employs a codebook of size $b_{kj}[\bar{t}],\bsp j\in {\mathcal N}_k[\bar{t}]$ bits for sub-band $j$ where $\sum_{j\in {\mathcal N}_k}b_{kj}[\bar{t}]=b_k[\bar{t}]$; $\{b_{kj}[\bar{t}]\}_{j\in {\mathcal N}_k}$ represents the per-sub-band budgets for user $k$ that are chosen by the user to maximize rate.

\textbf{Quantized channel state}: Given a budget of $b_k[\bar{t}]$ and a sub-band assignment ${\mathcal N}_k[\bar{t}]$, the actual post-quantization rate achieved by user $k$ at time $t$ where $\bar{t} \leq t \leq \bar{t} + 1,$ on sub-band $j\in {\mathcal N}_k[\bar{t}]$ is denoted by $\mu_{kj}^q[b_{kj}[\bar{t}],\mu_{ij}[t]]$. Note that the actual achievable rate is determined by the quantization budgets (along with the codebook of course), that are decided on the slower time-scale indexed by $\bar{t}$, as well as the true state of the channel at current time $t$.

\textbf{Network state}: The network state at time $\bar{t}$ is given by $M[\bar{t}]=\left(\{\pi_{mi}[\bar{t}]\}_{i=1}^K,\{q_{k}[\bar{t}]\}_{i=1}^K\right)$, which is a collection of channel distributions and queue lengths on the slower time-scale. In general, the feedback allocation and slow scheduling policies make allocation and assignment decisions, respectively, for $T$ time slots $\bar{t}< t < \bar{t} + 1$ based on state $M[\bar{t}]$.

\textbf{Expected rates and \emph{virtual} users}: Let $r_{kj}[\alpha_k[\bar{t}],b_{kj}[\bar{t}]] = {\mathbb E}_{\pi_{mi}[\bar{t}]}\left[\mu_{kj}^q[b_{kj}[\bar{t}],\mu_{ij}[t]]\right]$ denote the expected rate (through the course of $T$ time slots) for user $k$ on sub-band $j$. The total expected rate that is achieved by user $k$ given a sub-band assignment ${\mathcal N}_k[\bar{t}]$ and allocation $b_k[\bar{t}]$ is then given by $\sum_{j\in{\mathcal N}_k[\bar{t}]}r_{kj}[\alpha_k[\bar{t}],b_{kj}[\bar{t}]].$ We make an observation at this point that helps us simplify the presentation of the results. Since the rate is additive across sub-bands, and is a function of a band-independent channel gain, one may consider and analyze an equivalent \emph{virtual} system where the number of users is \emph{equal} to the number of sub-bands. This removes the dependence of the feedback allocation policy on the assignments ${\mathcal N}_k[\bar{t}]$. In other words, the equivalent system would consist of $L$ users assigned to $L$ sub-bands with feedback allocations $\{b_{i^{\prime}}[\bar{t}]\}_{i^{\prime}=1}^L$ and rates $r_{i^{\prime}}[\alpha_{i^{\prime}}[\bar{t}],b_{i^{\prime}}[\bar{t}]]$. As for the queue lengths, one can simply ``replicate'' the same queue length $q_k[\bar{t}]$ for all virtual users $k^{\prime}\in{\mathcal N}_k[\bar{t}]$, i.e., $q_{k^{\prime}}[\bar{t}] = q_k[\bar{t}],\bsp \forall k^{\prime}\in {\mathcal N}_k[\bar{t}]$. Once the optimal feedback allocation $\{b_{i^{\prime}}^*[\bar{t}]\}_{i^{\prime}=1}^L$ and virtual rates $r_{i^{\prime}}[b_{i^{\prime}}^*[\bar{t}]]$ are computed, we can map back to the original system by calculating the true rate for user $k$ as $\sum_{i^{\prime}\in{\mathcal N}_k[\bar{t}]}r_{i^{\prime}}[\alpha_{i^{\prime}}[\bar{t}],b_{i^{\prime}}^*[\bar{t}]]$.

Through the remainder of this paper barring the final simulations section, we study the equivalent system mentioned above where we have $L$ users assigned to $L$ sub-bands. Having defined all the ingredients of our OFDMA downlink network, we move on to the next section where we develop the feedback allocation policy that periodically makes decisions based on the network state.

\section{Throughput-optimal feedback allocation with slow scheduling}
\label{sec:tputopt}
In this section, we develop a feedback allocation (codebook size adaptation) protocol that when operated in conjunction with a given slow data scheduling policy, guarantees \textit{throughput-optimality}. This means that given an arrival rate vector $\boldsymbol{\lambda}$, if there exists any scheduling policy that can guarantee bounded expected queue sizes, then so can the proposed policy, which falls under the MaxWeight family of policies that was pioneered by Tassiulus and Ephremedis \cite{tasseph}.

As mentioned towards the end of the last section, we now have a virtual system with $L$ users assigned to $L$ sub-bands with feedback allocations $\{b_{k^{\prime}}[\bar{t}]\}_{k^{\prime}=1}^L$, rates $\{r_{k^{\prime}}[\alpha_{^{\prime}}[\bar{t}],b_{k^{\prime}}[\bar{t}]]\}_{k^{\prime}=1}^L$ and queues $\{q_{k^{\prime}}[\bar{t}]\}_{k^{\prime}=1}^L$. Through the remainder of this paper, until Section~\ref{sec:simulations}, we replace the index $k^{\prime}$ by $k$ for convenience, with the implicit understanding that we are dealing with virtual users. The feedback allocation policy is presented below.
\begin{algorithm}
\caption{MaxWeight feedback allocation with slow data scheduling}
\label{alg:1}
\begin{algorithmic}[1]
\WHILE{$t\geq 0$}
\IF{$t\bsp (\mbox{mod } T) = 0$}
    \STATE Set $\bar{t} = t$
    \STATE Solve
    \be
    \ba
    \{b_k^*[\bar{t}]\}=&\arg\max& \sum_{k=1}^L q_k[\bar{t}] r_{k}[\alpha_{k}[\bar{t}],b_{k}]\\    &\mbox{w.r.t.}&b_{k}\in\{0,1,\ldots,B\},\bsp \forall k\\
    &\mbox{s.t.}& \sum_{k=1}^L b_k\leq B.\label{jointFBSCH}
    \ea
    \ee
\ENDIF
\ENDWHILE
\end{algorithmic}
\end{algorithm}

A few remarks about the above algorithm are in order:

\noi \emph{(i)} \emph{Throughput-optimality}: The algorithm is throughput-optimal with respect to the space of policies that make feedback allocation and assignment decisions once every $T$ slots. This means that if any policy that makes feedback allocation and assignment decisions once every $T$ slots can stabilize a set of arrival rates $\{\lambda_k\}$, then so can the proposed policy in (\ref{jointFBSCH}). Let the region of rates that can be stabilized by the policy in (\ref{jointFBSCH}) be denoted by ${\mathcal V}$. \revbegin The above notion of throughput optimality for queueing systems has been used extensively in the literature (see \cite{tasseph,akrsvw,shroff,srikant} and references therein)\revend. We do not prove throughput-optimality as it follows from standard Lyapunov techniques that are well-established in the queueing literature \cite{akrsvw}.

\noi \emph{(ii)} \emph{Computational complexity}: While the optimization problem characterizes optimal performance, solving it exactly may be computationally prohibitive. In fact, a brute-force approach to solving (\ref{jointFBSCH}) would incur a complexity of ${\mathcal O}\left({B+L-1 \choose L-1}\right)$.\\

The final remark forms the basis for the remainder of this paper. The brute-force approach to solving (\ref{jointFBSCH}) is clearly infeasible from an implementation perspective. We take up the issue of complexity starting in Section~\ref{sec:optsol} and propose a host of computationally-efficient algorithms to solve the feedback allocation problem in (\ref{jointFBSCH}). We wish to highlight that all algorithmic developments can be applied to full-buffer (saturated) systems where scheduling schemes such as \textit{proportional fairness} become applicable. This is because most schedulers of interest solve a weighted sum-rate maximization problem at each instant \cite{stolyar_noqueues}.

\section{Optimal feedback allocation through dynamic programming}
\label{sec:optsol}

In Section~\ref{sec:tputopt}, we have established that for queue stability, we are interested in solving the following online weighted sum-rate maximization problem
\begin{equation}
\begin{array}{rcl}
\max_{\{ b_k\}\in{\mathcal B}}& \sum_{k=1}^L q_k[\bar{t}] r_{k}[\alpha_{k}[\bar{t}],b_{k}].\label{sumrate_online}
\ea
\end{equation}
The form of the functions $\{r_{k}[\alpha_{k}[\bar{t}],b_{k}]\}$ would of course depend on the underlying physical system and is intimately connected to the computational complexity of the problem. In fact, for complex modulation/coding schemes the function might only be available as a look-up table. While the optimization problem characterizes optimal performance, solving it exactly may be computationally prohibitive. Thus, the focus of this paper becomes algorithmic. We propose novel solutions to (\ref{sumrate_online}) through Sections~\ref{sec:optsol} and \ref{sec:redcomp} that explore the natural tradeoffs between accuracy, complexity and the structure of the weighted sum-rate function. We start by showing that by using Dynamic Programming, the exact solution can be obtained in pseudo-polynomial time.
\begin{theorem}
\label{thm_stability}
The online resource allocation problem (\ref{sumrate_online}) can be solved exactly in time ${\mathcal O}\left(LB^2\right)$.
\end{theorem}
\begin{proof} Order the users arbitrarily. \revbegin We choose to work with the existing order w.l.o.g. For any given arbitrary weights $\{w_i\},\bsp w_i > 0$, define $A(i,j)\eqdef w_i r_{k}[\alpha_{i},j]$ to be the weighted rate for user $i$ given we allocate $j$ bits to this user and define $R(k,b)\eqdef \max_{\sum_{i=1}^k b_i\leq b,\bsp b_i\in{\mathbb N}_0} \sum_{i=1}^k w_ir_i[\alpha_i,b_i]$ to be the maximum weighted sum-rate if we have $b$ bits to allocate amongst the first $k$ users with $R(0,b)=0$\revend. It follows that $R(1,b)=A(1,b),b=0,\ldots,B$. We can write a recursion $R(k,b)=\max_{j=0,\ldots,b} \left\{ R(k-1, b-j) + A(k,j) \right\}$. The optimality of this recursion can be established using standard induction arguments similar to the two-dimensional knapsack problem \cite{kleinberg}. This rule gives rise to a table with a total of $L(b+1)$ elements. In order to compute element $(k,b)$ in the table, using our recursion, we incur a complexity of ${\mathcal O}(b+1)$. Hence, the total complexity can be calculated as $\sum_{k=1}^L\sum_{b=0}^B (b+1)=L \sum_{b=0}^B (b+1) = L\frac{(B+1)(B+2)}{2}={\mathcal O}(LB^2)$.
\end{proof}
Thus, we have proposed an exact solution using dynamic programming, which has pseudo-polynomial\footnote{An algorithm has pseudo-polynomial complexity if its running time is a polynomial in the size of the input in unary. The size of the input to (\ref{sumrate_online}) in unary at most $LBA_{\max}+B={\mathcal O}(LB)$ where $A_{\max}=\max_{(i,j)}A(i,j)$.} complexity ${\mathcal O}\left(LB^2\right)$ and which is applicable to \emph{any} type of weighted sum-rate function. \revbegin Therefore, in contrast to the joint power-control and scheduling problems in \cite{tan1,tan2,huang} and owing to the discrete nature of the feedback allocation problem we consider in (\ref{sumrate_online}), we do not require any special channel-induced properties of the objective function such as those imposed on its partial derivatives in Lemma $2$ of \cite{tan1}, in order to find an optimal solution.\revend

\revbegin Note that $R(K,B)$ in Theorem~\ref{thm_stability} with $w_i = q_i,\bsp \forall i$, is equal to (\ref{sumrate_online}) and dynamic programming essentially gives us a technique to compute $R(K,B)$ by solving smaller sub-problems. The following toy example with $K=2$ users and a total bit budget of $B = 2$ bits illustrates a typical series of computations en route to calculating $R(2,2)$.

\emph{Example (Dynamic programming)}: Order the two users arbitrarily, say user $1$ first followed by user $2$. Then, initialize the following weighted rates appropriately for $b = 0,1,2$,
\begin{equation*}
\begin{array}{rcl}
R(1,b) = A(1,b)&:& \mbox{when user $1$ is allocated $b$ bits}\\
R(2,b) = A(2,b) &:& \mbox{when user $2$ is allocated $b$ bits}.\\
\ea
\end{equation*}
Once initialized, we then compute value $R(2,1) = \max\{R(1,1)+A(2,0),R(1,0)+A(2,1)\} = \max\{R(1,1),A(2,1)\}$. Finally, we calculate
\begin{equation*}
\ba
R(2,2) &= &\max\{R(1,2)+A(2,0),R(1, 1)+A(2,1),\\
&&R(1,0)+A(2,2\}\\
 &= &\max\{R(1,2),R(1,1)+A(2,1),A(2,2)\},
 \ea
 \end{equation*}
the desired optimal weighted sum-rate with two users and two bits.\revend

To understand the computational requirements in the context of a real-world system, consider the LTE example that was presented in the introduction to this paper. Here, we had the following parameters: $L=50$, $K=50$ with $4$-bit modulation/coding tables at the mobiles. To model the limited feedback constraint, let $B = 4cL,\bsp c\in\{1,2,...,K-1\}$, which has an intuitive interpretation of being able to provide full feedback to at most $c$ users; $c = K$ represents no constraint on feedback resources for this setting. Then, a feedback bandwidth of $c = \frac{K}{4}$ corresponds to a complexity of roughly $7\times 10^{11}$ operations, which is clearly quite daunting. Thus, while the dynamic programming approach is indeed viable for sufficiently small systems, we require algorithms with faster running times that might be less accurate. This forms the focus of the remainder of this paper.

\section{Reduced-complexity resource allocation}\label{sec:redcomp}
In this section, we show that if the weighted sum-rate functions have additional structure, we can develop faster algorithms. As is often done for computationally hard problems, one seeks efficient but potentially sub-optimal algorithms, but then proves lower bounds on the performance. In this vein, we develop more computationally efficient algorithms that approximately solve (\ref{sumrate_online}), and provide theoretical lower bounds on their performance. The long-term performance of these approximate algorithms in achieving queue stability is characterized by Theorem~\ref{approxstability} below. The proof is omitted as these are well-known results in queuing systems.

We say that an algorithm is a $\gamma$-approximation, $\gamma\in(0,1]$, to (\ref{jointFBSCH}) if it provides a solution $\{b^{alg}_k\}$ such that $
\sum_k q_k[\bar{t}] r_{k}[\alpha_{k}[\bar{t}],b_{k}^{alg}]\geq  \gamma \max_{\{b_k\}\in{\mathcal B}} \sum_k q_k[\bar{t}] r_{k}[\alpha_{k}[\bar{t}],b_k]$. The following theorem is a generalization of the original result by Tassiulus and Ephremedis \cite{tasseph}. It essentially states that local approximation is consistent with the long-term objectives we consider.
\begin{theorem}\label{approxstability}
If $\boldsymbol{\lambda}\in\{\gamma \boldsymbol{\nu} : \boldsymbol{\nu} \in {\mathcal V}\},\gamma\in(0,1]$, then a $\gamma$-approximation to the per-instant scheduling rule $
\{b^*[t]\}=\arg \max_{\{b_k\}\in{\mathcal B}} \sum_k q_k[\bar{t}] r_{k}[\alpha_{k}[\bar{t}],b_k]$ stabilizes the system.
\end{theorem}
Recall from remark {\em (i)} in Section~\ref{sec:tputopt} that ${\mathcal V}$ represents the region of rates that are stabilizable by Algorithm~\ref{alg:1}. The theorem essentially states that for queueing systems: If we calculate a $\gamma$-approximate solution to (\ref{sumrate_online}) in every time slot, one can achieve a $\gamma$-fraction of the \revbegin throughput \revend region. This result paves the way for the design of computationally-efficient algorithms, by constructing approximations to (\ref{sumrate_online}).

In Section \ref{submodularsection}, we consider weighted sum-rate functions that are non-decreasing and sub-modular in the feedback bit allocation. In short, sub-modularity refers to \textit{diminishing returns} with respect to the allocation of resources. This is a property that is exhibited quite frequently by wireless systems in general, e.g., point-to-point capacity scales logarithmically in transmit power, achievable rates in multiple antenna precoding systems exhibit diminishing returns in the size of the codebook \cite{love_RVQ1,love_grass}, etc. In the developments that follow, we exploit this property in order to propose a {\em greedy} feedback allocation algorithm that has complexity ${\mathcal O}((B+L)\mbox{log}_2 L)$ with approximation factor of $\left(1-\frac{1}{e}\right)$. Our main contributions are contained in Lemma~\ref{SM_our_function} and Theorem~\ref{greedycomp}.

In Section~\ref{mimowaterfilling}, we focus on a class of weighted sum-rate functions that arise in downlink scenarios where the base station is equipped with multiple antennas and performs transmit beamforming with quantized beamformer feedback. This is a popular transmission strategy that been extensively researched \cite{gesbert_MIMO,raghavan,love_RVQ1,love_RVQ2} and adopted into standards such as W-CDMA \cite{WCDMAbook} and LTE \cite{ltebook1}. We show that for this choice of physical layer scheme, the weighted sum-rate maximization problem in (\ref{sumrate_online}) is sub-modular for certain types of beamformer quantizers. We illustrate the improvement in computational performance by using the LTE example from the introduction.

\subsection{Reduced-complexity resource allocation through sub-modularity}\label{submodularsection}
We begin this section with a quick primer on sub-modular optimization (summarized from \cite{goundan,calinescu,vondrak}) that will be useful for our purposes. In keeping with the literature, the approach pursued in this section is graph theoretic in contrast to the rest of this paper. A sub-modular function is defined as follows: Let $E$ be a finite set and $2^E$ represent all its subsets. Then, $F:2^{E}\rightarrow{\mathbb R}_{+}$ is a \textit{non-decreasing}, \textit{normalized}, \textit{sub-modular} function if $F(\emptyset)=0$ (normalized), $F(A)\leq F(B)$ when $A\subseteq B\subseteq E$ (non-decreasing) and if $F(A\cup \{e\})-F(A)\geq F(B\cup \{e\})-F(B),\bsp \forall A\subseteq B\subseteq E$ and $e\in E\setminus B$ (sub-modular).

The following property of sub-modular functions is useful for our analysis.
\begin{lemma}\label{additivity}
If $F_n(\cdot),\bsp n=1,\ldots,N$, are sub-modular on set $E$, then $\sum_{n=1}^N w_n F_n(A),\bsp A\subseteq E$ is a sub-modular function for $w_n\geq 0,\forall n$.
\end{lemma}
Having provided the definition of sub-modularity along with a useful property, we now introduce the kinds of constraint sets that are typically considered in the context of sub-modular optimization: \textit{(i)} A set system $(E,{\mathcal I})$ where $E$ is a finite set and ${\mathcal I}$ is a collection of subsets of $E$ is called an \textit{independence system} if $\emptyset\in{\mathcal I}$ and satisfies if $A\subseteq B$ for $B\in {\mathcal I}$, then $A\in {\mathcal I}$. \textit{(ii)} An independence system is called a \textit{matriod} if it satisfies the following additional property; if $A,B\in {\mathcal I}$ and $|A|< |B|$, then there exists $e\in B\setminus A$ such that $A\cup \{e\}\in {\mathcal I}$. (iii) Set ${\mathcal I}$ is a \emph{uniform matroid} if ${\mathcal I}=\{F\subseteq E:|F|\leq k\}$ for $k\in{\mathbb N}$.

The optimization problem that has been considered in the context of sub-modular functions and independence systems is
\be
\bs
F^{*}=\max_{A\in{\mathcal I},A\subseteq E} F(A).\label{submod}
\es
\vspace{-0.1cm}
\ee
Since many NP-hard problems can be reduced to a sub-modular function maximization over an independence system, significant research has focused on developing efficient approximation algorithms. In particular, the performance of the greedy algorithm in solving special cases of (\ref{submod}) has been extensively studied. Nemhauser et al. \cite{nemhauser} considered problem (\ref{submod}) over uniform matroids and showed that the greedy algorithm provides a $(1-\frac{1}{e})$ approximation factor for this special case. Please refer to Goundan et al. \cite{goundan}, Calinescu et al. \cite{calinescu} and Vondrak \cite{vondrak} for a summary of related results on sub-modular function optimization over other families of constraint sets. The greedy algorithm is presented later in the section in the context of our specific feedback allocation problem.

\textbf{Sub-modularity in feedback allocation}: We now show that the optimal bit allocation problem in (\ref{sumrate_online}) may by posed as a sub-modular maximization over a uniform matroid when the rates exhibit sub-modularity. Let ${\mathcal G} = (U, V, E)$ be a bipartite graph where $U$ contains $L$ \textit{user nodes} and $V$ contains $B$ \textit{bit nodes}, both ordered arbitrarily, i.e., $|U|=L$ and $|V|=B$. Let $E$ contain the set of all edges $E=\{e_{kb}:i=1,\ldots,L\mbox{ and }j=1,\ldots,B\}$. Given $A\subseteq E$, we define $|A|_i\eqdef|\{e_{kb}\in A:k=i\}|$ to represent the number of bits allocated to user $i$, i.e., $|A|_i=b_i$. The independence system we are interested in is ${\mathcal I}=\{A\subseteq E:|A|\leq B\}$ where $B$ is the total bit budget. By definition, ${\mathcal I}$ is a uniform matroid and furthermore, ${\mathcal I}$ is the set of \textit{all} valid allocations since if $A\in {\mathcal I}$, then $\sum_{k=1}^L b_k=\sum_{k=1}^L |A|_k\leq B$ and if $A\not\in {\mathcal I}$, then $\sum_{k=1}^L b_k=\sum_{k=1}^L |A|_k=|A|>B$.  Now the weighted sum-rate maximization problem in (\ref{sumrate_online}) in time slot $\bar{t}$ may be re-written as
$$
\begin{array}{rll}
&\max_{\{ b_k\}\in{\mathcal B}}&\sum_{k=1}^L q_k r_{k}[\alpha_{k}[\bar{t}],b_{k}]\\
\equiv&\max&\sum_{k=1}^L q_k r_{k}[\alpha_{k}[\bar{t}],b_{k}]-r_{k}[\alpha_{k}[\bar{t}],0]\\ \label{sumrate_online_submodular}
&\mbox{s.t. }&b_k=|A|_k,\bsp \sum_{k}|A|_k\leq B,\bsp A\subseteq E\\
=&\max_{A\in{\mathcal I}}&\sum_{k=1}^L q_kr_{k}[\alpha_{k}[\bar{t}],|A|_k]-r_{k}[\alpha_{k}[\bar{t}],0].
\ea
$$
The following result becomes immediate.
\begin{lemma}\label{SM_our_function}
If the function $r_{k}[\alpha_{k},b_{k}]$ is non-decreasing and sub-modular in the bit allocation $b_k=|A|_k,\bsp A\subseteq E$ for all users $k=1,\ldots,L,$ and channel states $\alpha_k\in\Omega$, then $\sum_{k=1}^L q_k r_{k}[\alpha_{k},|A|_k]-r_{k}[\alpha_{k},0]$ is a normalized, non-decreasing, sub-modular function on set $E$ for all channel states $\{\alpha_k\}\in{\mathcal Omega}^L$. 
\end{lemma}
\begin{proof}
The result follows from Lemma \ref{additivity}.
\end{proof}
Hence, the greedy algorithm can be used to solve the optimal bit allocation problem in (\ref{sumrate_online}) with approximation factor $\left(1-\frac{1}{e}\right)$. The greedy algorithm for the specific case of our bit allocation problem in time slot $\bar{t}$ can be written as follows where $u_{k}(b_k)\eqdef q_k\left(r_{k}[\alpha_{k},b_{k}+1]-r_{k}[\alpha_{k},b_{k}]\right)$ denote the increase in rate or marginal utility if user $k$ is given one extra bit.
\begin{algorithm}
\caption{Greedy feedback allocation}
\label{alg:2}
\begin{algorithmic}[1]
    \STATE Set $b_k=0,\forall k$, which is essentially a bit counter for each user
    \STATE Compute marginal utilities $\{u_{k}(b_k)\}$.
    \WHILE{$\sum_k b_k \leq B$}
    \STATE Sort this list of marginal utilities.
    \STATE Assign a bit to user $k^*$ who is on top of this list.
    \STATE Update $b_{k^*}=b_{k^*}+1$ and re-compute $u_{k^*}(b_{k^*})$
    \ENDWHILE
\end{algorithmic}
\end{algorithm}

We end this section by investigating the complexity of the above algorithm in the following theorem.

\begin{theorem} \label{greedycomp}
The greedy algorithm approximates the optimal bit allocation problem in (\ref{sumrate_online}) to within a factor of $\left(1-\frac{1}{e}\right)$ while incurring complexity ${\mathcal O}((B+L)\mbox{log}_2 L)$.
\end{theorem}
\begin{proof}
Step $2$ of this algorithm incurs complexity ${\mathcal O}(L\mbox{log}_2 L)$ for the first iteration $b=1$. Subsequently, every re-sort in Step $3$ costs ${\mathcal O}(\mbox{log}_2 L)$ with a maximum of $B$ such re-sorts. Thus, the total complexity is ${\mathcal O}((B+L)\mbox{log}_2 L)$. For the proof of the approximation factor, please refer to Nemhauser et al. \cite{nemhauser}.\end{proof}

In the context of the LTE example introduced earlier, this means that by exploiting the sub-modular structure in the rates, we reduce the complexity from $7\times 10^{11}$ to $15\times 10^3$ operations. \revbegin Before we move on to the next section, we provide an example of a common wireless setting where sub-modularity is exhibited. Consider a traditional point-to-point single antenna wireless link with a $b$-bit modulation-coding table at the receiver. The modulation-coding table is constructed as follows. Given a non-negative real number in the interval $[0,\sigma],\bsp \sigma >> 0$, we uniformly partition the interval into $2^{b}$ sub-intervals and implement the quantization scheme ${\left\lfloor x\right\rfloor}_Q = \frac{i\sigma}{2^b},i\frac{\sigma}{2^b}\leq x < (i+1)\frac{\sigma}{2^b},\bsp i = 0,1,\ldots,2^b-1.$ Then, for any fixed position-dependent gain of $\alpha$, the achievable rate of the system in a fading environment can be written as
\be
r[\alpha,b] = {\mathbb E}_{h}[\log_2(1+{\left\lfloor\sqrt{\alpha} |h|^2\right\rfloor}_Q)],\label{rate_SISOquant}
\ee
where $|h|^2$ is a truncated $\exp(1)$ random variable that has a maximum value of $\sigma >> 0$. The probability density function for such a random variable is given by $f_{|h|^2}(x) = C(\sigma)\exp(-x)$ where $C(\cdot)$ is a normalization factor. Note that this example considers a traditional continuous fading model. One may obtain its discrete version thereby conforming with our system model, by sampling the support $[0,\sigma]$. Thus, the rate expression in (\ref{rate_SISOquant}) may be treated as an approximation that becomes increasingly accurate as we discretize the support more finely. For the case $\alpha  = 1$, the rate (\ref{rate_SISOquant}) can be explicitly computed as
\begin{equation*}
\begin{small}
\ba
r[1,b] &=& C(\sigma)\sum_{i = 0}^{2^b-1} \log_2\left(1+\frac{i\sigma}{2^b}\right) \int_{i\frac{\sigma}{2^b}}^{(i+1)\frac{\sigma}{2^b}} \exp(-x)dx\\
 &=& \left[1-\exp\left(-\frac{\sigma}{2^b}\right)\right]\sum_{i = 0}^{2^b-1} \log_2\left(1+ \frac{i\sigma}{2^b}\right)\exp\left(-i\frac{\sigma}{2^b}\right).
\ea
\end{small}
\end{equation*}
Setting $l[j,b] = \log_2\left(1+\frac{j\sigma}{2^b}\right)$, the normalized incremental gain with one extra bit can be calculated as
\be
\begin{small}
\begin{array}{rl}
&C(\sigma)^{-1}(r[1,b+1] - r[1,b]) \\
=& \left[1-e^{-\frac{1}{2}\frac{\sigma}{2^b}}\right]\sum_{i = 0}^{2\times 2^b-1} l[i,b+1]e^{-i\frac{1}{2}\frac{\sigma}{2^b}}-\left[1-e^{-\frac{\sigma}{2^b}}\right]\\
&\sum_{i = 0}^{2^b-1} l[i,b]e^{-i\frac{\sigma}{2^b}}\\
=&\left[1-e^{-\frac{1}{2}\frac{\sigma}{2^b}}\right]\left[\sum_{j = 0}^{2^b-1} l[2j,b+1]e^{-j\frac{\sigma}{2^b}}+\sum_{j = 0}^{2^b-1} l[2j+1,\right.\\
&\left.b+1]e^{-(2j+1)\frac{\sigma}{2^{b+1}}}\right]- \left[1-e^{-\frac{\sigma}{2^b}}\right]\left[\sum_{i = 0}^{2^b-1} l[i,b]e^{-i\frac{\sigma}{2^b}}\right]\\
&\mbox{by splitting odd and even terms}\\
=&\left[e^{-\frac{1}{2}\frac{\sigma}{2^b}}-e^{-\frac{\sigma}{2^b}}\right]\left[\sum_{j = 0}^{2^b-1}e^{-j\frac{\sigma}{2^b}} \log_2\left(1+\frac{0.5}{\frac{2^b}{\sigma}+j}\right)\right].\label{uniform_quant_loss}
\ea
\end{small}
\ee
Through simple numerical enumeration, one may confirm that the $1$-bit rate gain given above in (\ref{uniform_quant_loss}) decreases over realistic bit sizes of $b\in\{1,2,\ldots,25\}$ and hence, $r[1,b]$ is a sub-modular function. With a little more algebra, one may derive a similar result for the general case with any arbitrary, non-negative, position-dependent gain $\alpha$.

In the next section, we provide another example of a wireless system that exhibits sub-modularity. In particular, we consider a class of multiple antenna wireless links and solve (\ref{sumrate_online}) in the context of these systems.
\revend

\subsection{Reduced-complexity resource allocation for MISO systems }\label{mimowaterfilling}

When the user rates $r_{k}[\alpha,b]$ are sub-modular in the bit allocation $b$ in every channel state $\alpha\in\Omega$, we use the greedy algorithm in Section \ref{submodularsection} to approximately solve the online feedback allocation problem in (\ref{sumrate_online}) with complexity ${\mathcal O}((B+L)\mbox{log}_2 L)$. In this section, we show that $2\times 1$ MISO quantized transmit diversity systems exhibit sub-modular expected rates bringing into use the results from the previous section. Furthermore, \revbegin  in the context of these specific transmission schemes\revend, we develop an approximation algorithm based on convex relaxations with a further-reduced complexity of ${\mathcal O}(L\mbox{log}_2 L)$ and an approximation guarantee of $\frac{1}{2}$ for typical operational signal-to-noise ratios (SNR). Thus, aside from the usual impact on precision that is typically omitted from running time calculations, the running time of our algorithm no longer depends on the feedback budget $B$. In the example above, the running time is reduced even further from $15\times 10^3$ operations to roughly $300$ operations.

We begin this section by investigating the effects of limited feedback on the aforementioned class of MISO systems. It is well-known that the instantaneous SNR for a classical $2\times 1$ single-stream beamforming MISO link is given by $\overline{\text{SNR}}(\alpha)||{\bf h}||^2$ where $\overline{\text{SNR}}(\alpha) = \frac{P\alpha}{N_o}$, $P$ is the transmit power, $N_o$ is the noise power and ${\bf h}=[h_1\bsp h_2]^T,\bsp h_i\in{\mathbb C}$ represents the MISO channel with zero mean, unit variance complex Gaussian entries. As with the example in the previous section, the analytical rate expressions in this section are derived for continuous vector channels, which are increasingly accurate approximations as we sample the support ${\mathbb C}^2$ more finely. Recall from Section~\ref{sec:sysmodel} that $\alpha\in\Omega$ models the effects of large-scale fading. To achieve this maximum instantaneous SNR, the user requires perfect feedback of the channel vector ${\bf h}$. However, feedback in realistic systems is imperfect due to limited feedback budgets and quantization, the primary motivation for this work. We assume that the channel vector $\bf h$ is quantized using the popular \textit{Random Vector Quantization} (RVQ) technique \cite{love_RVQ1,love_RVQ2}. This technique is briefly reviewed in the next section when we present simulation results. Recent results \cite{raghavan,love_RVQ1,love_RVQ2} bound (upper and lower) the loss in rate due to quantization of $\bf h$ when using RVQ codebooks. In particular, both upper and lower bounds on the rate loss due to quantization for user $k$ take the form $c(\alpha_k)2^{-b_k}$ for some $c(\alpha_k)>0$. Motivated by these results, we assume that the post-quantization rate for user $k$ in the presence of large-scale fading takes the form $r_{k}[\alpha_{k},b_{k}] = {\mathbb E}\left[\log_2(1+\overline{\text{SNR}}_k||{\bf h}||^2)\right]-\left({\mathbb E}\left[\log_2(1+\overline{\text{SNR}}_k||{\bf h}||^2)\right]- {\mathbb E}\left[\log_2(1+\overline{\text{SNR}}_k|h|^2)\right]\right)2^{-b_k}$, where we have omitted the dependence on $\bar{t}$ for brevity. We validate the use of the above approximation through numerical testing in the next section for many values of $\alpha_k$ from a typical operational range in a wireless system.

Thus, the optimization in (\ref{sumrate_online}) for the $2\times 1$ MISO case takes the specific form
\begin{equation}
\begin{small}
\begin{array}{rcl}
\mbox{\hspace{-0.25cm}}\max_{\{b_k\}\in{\mathcal B}}&\mbox{\hspace{-0.4cm}} \sum_{k=1}^L q_k\left[ \beta_2(\overline{\text{SNR}}_k)\mbox{\hspace{-0.1cm}}-\mbox{\hspace{-0.1cm}}\left(\beta_2(\overline{\text{SNR}}_k) - \beta_1(\overline{\text{SNR}}_k)\right)2^{-b_k}\right],\label{DP_PHY}
\ea
\es
\end{equation}
where $\text{SNR}_k=\text{SNR}(\alpha_k)$ for short, $\beta_1(\overline{\text{SNR}}) = {\mathbb E}\left[\log_2(1+\overline{\text{SNR}}|h_1|^2)\right]$ and $$\beta_2(\overline{\text{SNR}}) = {\mathbb E}\left[\log_2(1+\overline{\text{SNR}}||{\bf h}||^2)\right]$$ denote the one-tap and two-tap expected rates, respectively, for a Rayleigh fading channel with the given SNR. \\

\noi \textbf{Relaxation and approximation guarantees}: Through the remainder of the section, we develop an approximation algorithm to solve (\ref{DP_PHY}) in closed-form while incurring a complexity of ${\mathcal O}(L\mbox{log }_2 L)$\footnote{We recognize that there is an additional storage cost of ${\mathcal O}(\mbox{log }B)$.}. We provide an approximation guarantee of $\frac{1}{2}$.

\begin{theorem}
\label{thm_DPrelaxation}
Consider the following continuous relaxation of (\ref{DP_PHY}) formed by replacing the discrete set ${\mathcal B}$ with its natural continuous extension and dropping terms that are independent of the variables $\{b_k\}$:
\begin{equation}
\bs
\{b_k^*\}=\arg\min_{\sum_{k}b_k\leq B,\bsp b_k\in{\mathbb R}_{+}} \sum_{k=1}^L q_k \beta_1(\overline{\text{SNR}}_k)2^{-b_k}.\label{waterfilling_objective}
\es
\end{equation}
The solution to this relaxation is $b_k^*=\left[-\log_2\left(\frac{\eta^*}{q_k (\log\bsp 2)}\frac{1}{\beta_1\left(\overline{\text{SNR}}_k\right)}\right)\right]^+$, where $\eta^*$ is chosen such that $\sum_{k}b_k^*=B$ and $[x]^+=\max\{x,0\}$.
\end{theorem}
\begin{proof}See Appendix~\ref{proofs_MIMO}.\end{proof}
Next, we comment on the complexity of computing the above fractional solution.
\begin{theorem}
\label{thm_DPcomplexity}
Computing the above solution in Theorem~\ref{thm_DPrelaxation} incurs a complexity of ${\mathcal O}(L\mbox{log}_2 L)$.
\end{theorem}
\begin{proof} See Appendix~\ref{proofs_MIMO}.\end{proof}
The following lemma states that weighted sum-rate function in (\ref{DP_PHY}) is non-decreasing and sub-modular on set $E=\{e_{kb}:i=1,\ldots,L\mbox{ and }b=1,\ldots,B\}$, thereby allowing us to connect and compare the results in this section with those in the previous section on sub-modular functions. The proof is standard in the literature on sub-modular functions and follows from the fact that the fractional relaxation of the weighted sum-rate function is concave in $\{b_k\}$ over the domain $\{[0,B]^K:\sum_{k}b_k\leq B\}$. It is hence omitted.
\begin{lemma}
The weighted sum-rate function in (\ref{DP_PHY}) where $b_k=|A|_k,\bsp A\subseteq E$, $E=\{e_{kb}:i=1,\ldots,L\mbox{ and }b=1,\ldots,B\}$ is non-decreasing and sub-modular on this set $E$.
\end{lemma}
Comparing the results in Theorems~\ref{greedycomp} and~\ref{thm_DPcomplexity}, we see that by assuming less about the exact form of the communication system, we are incurring an added complexity cost of ${\mathcal O}(B \mbox{log}_2 L)$, while providing a system-independent approximation guarantee of $(1-\frac{1}{e})$.

Once we solve for $b_k^*$, we apply a floor operation in order to enforce the integer constraints, i.e., we set $b_{k,INT}^*=\left\lfloor b_k^* \right\rfloor\mbox{ if }b_k^*\geq 1$ and $b_{k,INT}^*=0\mbox{ if }b_k^*< 1$. This leads us to the task of quantifying loss due to integrality, which we do next. We consider two cases: For $b_k^* \geq 1$, we have that
\be
\begin{array}{rl}
&\frac{\beta_2(\overline{\text{SNR}}_k)(1-2^{-b_{k,INT}})+\beta_1(\overline{\text{SNR}}_k)2^{-b_{k,INT}}}{\beta_2(\overline{\text{SNR}}_k)(1-2^{-b_{k}^*})+\beta_1(\overline{\text{SNR}}_k)2^{-b_{k}^*}}\\
\geq & \frac{\beta_2(\overline{\text{SNR}}_k)(1-2^{-b_{k}^*+1})+\beta_1(\overline{\text{SNR}}_k)2^{-b_{k}^*-1}}{\beta_2(\overline{\text{SNR}}_k)(1-2^{-b_{k}^*})+\beta_1(\overline{\text{SNR}}_k)}\\ &\mbox{ since $b_{k}^*-1\leq b_{k,INT}^*\leq
b_{k}^*+1$}\\
\geq &  \frac{1}{2}\frac{\beta_1(\overline{\text{SNR}}_k)2^{-b_{k}^*}}{\beta_1(\overline{\text{SNR}}_k)2^{-b_{k}^*}}\mbox{ since $1 \leq b_k^*< \infty$}\\\label{bkgreaterone}
=&\frac{1}{2}.
\ea
\ee
Similarly for $b_k^* < 1$ and $b_{k,INT}=0$, we have that
\be
\begin{array}{rl}
&\frac{\beta_2(\overline{\text{SNR}}_k)(1-2^{-b_{k,INT}})+\beta_1(\overline{\text{SNR}}_k)2^{-b_{k,INT}}}{\beta_2(\overline{\text{SNR}}_k)(1-2^{-b_{k}^*})+\beta_1(\overline{\text{SNR}}_k)2^{-b_{k}^*}}\\
\geq & \frac{\beta_1(\overline{\text{SNR}}_k)}{\frac{1}{2}\beta_2(\overline{\text{SNR}}_k)+\beta_1(\overline{\text{SNR}}_k)}\mbox{ since $b_k^*<1$}\\
=&\frac{1}{\frac{1}{2}\frac{\beta_2(\overline{\text{SNR}}_k)}{\beta_1(\overline{\text{SNR}}_k)}+1}.\label{bklesserone}
\ea
\ee
From (\ref{bkgreaterone}) and (\ref{bklesserone}), we can compute the approximation factor as
\be
\min\left\{\frac{1}{2},\frac{1}{\frac{1}{2}\max_k \left\{\frac{\beta_2(\overline{\text{SNR}}_k)}{\beta_1(\overline{\text{SNR}}_k)}\right\}+1}\right\}.\label{approxfactorpresim}
\ee
Thus, the approximation factor critically depends on the ratio $\frac{\beta_{2}\left(\overline{\text{SNR}}_k\right)}{\beta_{1}\left(\overline{\text{SNR}}_k\right)}$, which essentially represents the rate gain due an extra tap or antenna.
\begin{figure}[hbt]
\centering
\includegraphics[scale=0.55]{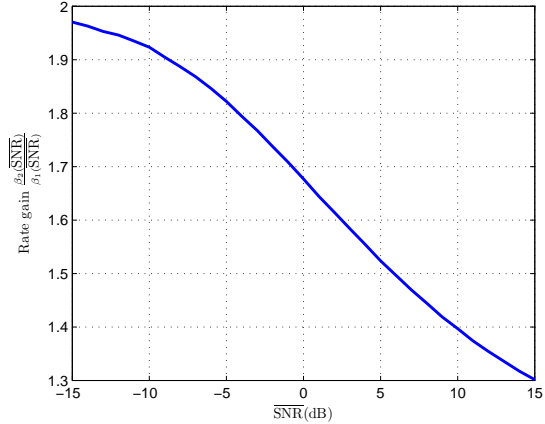}
\caption{Rate gain due to the addition of an extra antenna as a function of $\overline{\text{SNR}}$.}
\label{fig:approxfactor}
\end{figure}
In Fig.~\ref{fig:approxfactor}, we numerically compute $\frac{\beta_{2}\left(\overline{\text{SNR}}_k\right)}{\beta_{1}\left(\overline{\text{SNR}}_k\right)}$ for a typical cellular range of $-15$dB to $15$dB and see that $\frac{\beta_{2}\left(\overline{\text{SNR}}_k\right)}{\beta_{1}\left(\overline{\text{SNR}}_k\right)}\leq 2$ over this range. Combining the results in Fig.~\ref{fig:approxfactor} with (\ref{approxfactorpresim}), we can conclude that the proposed relaxation/rounding algorithm has an approximation factor of $\frac{1}{2}$.

\revbegin In summary, the two proposed algorithms exploit the structure of the feedback allocation problem in settings such MISO with quantized beamforming, to deliver lower complexity than the optimal dynamic programming approach accompanied by guarantees on accuracy. Note that the accuracy guarantees, namely, $(1-\frac{1}{e})$ for the greedy algorithm and $\frac{1}{2}$ for the convex program are independent of any system parameters such as channel statistics, total bit budget $B$, etc., and are hence, a clear measure of worst-case performance. We now move on to numerical simulations in the next section, which helps us understand the actual performance against the backdrop of these worst-case guarantees.\revend

\section{Numerical Simulations}\label{sec:simulations}
In this section, we evaluate the performance of the greedy feedback allocation algorithm in a MISO downlink network. The simulations serve as a proof of concept for the proposed dynamic feedback allocation approach. As the baseline case, we introduce a static equal allocation algorithm that we describe in detail below along with the rest of the simulation setup. Note that we now revert back (from the virtual system with $L$ users) to the original system with $K$ users, i.e., the indices $k=1,2,\ldots,K$, now track actual users.

\textbf{Number of users, OFDMA bands and data scheduling policy}: There are $K=4$ users in a $10$MHz system with a total of $L=8$ OFDMA sub-bands. Since the focus of these simulations (and this paper) is primarily on quantifying the gains of dynamic feedback allocation, the users are assigned equal amounts of spectrum for data transmission at the beginning of the communication epoch \emph{that do not change with time}, i.e., user $i$ is always assigned to bands $\{2i-1,2i\}$.

\textbf{Small-scale fading, average user SNRs and traffic model}: The users are stationary and have fixed average SNRs through the entire epoch of communication. We consider two average SNR profiles - \emph{(i)} Large asymmetry with average SNRs $10\log_{10}(\overline{\text{SNR}}_1[\bar{t}]) = -10$dB, $10\log_{10}(\overline{\text{SNR}}_2[\bar{t}]) = -8$dB, $10\log_{10}(\overline{\text{SNR}}_3[\bar{t}]) = 10$dB, $10\log_{10}(\overline{\text{SNR}}_4[\bar{t}]) = 10$dB and \emph{(ii)} Nearly symmetric with average SNRs $10\log_{10}(\overline{\text{SNR}}_1[\bar{t}]) = -1$dB, $10\log_{10}(\overline{\text{SNR}}_2[\bar{t}]) = -1$dB, $10\log_{10}(\overline{\text{SNR}}_3[\bar{t}]) = 1$dB, $10\log_{10}(\overline{\text{SNR}}_4[\bar{t}]) = 1$dB. Asymmetric profiles are of interest because this is the regime where dynamic allocation would arguably have most value. The small-scale fading channel realizations $\bf h$ in each slot are generated according to the standard complex Gaussian distribution. The arrivals are assumed to be deterministic and symmetric with rates $\lambda_k = \lambda,\bsp \forall k$.

\textbf{Feedback budget and baseline equal allocation}: The feedback budget is set to $B = 12$ bits. The baseline algorithm allocates an equal number of bits to each user, i.e., $b_k = 3,\bsp \forall k$. Each user in turn distributes these three bits as follows - two bits to the first sub-band it is assigned and one bit to the second. In other words, the per sub-band allocation for user $k$ is $b_{k1} = 2$ and $b_{k2} = 1$. The allocation is changed every $T=10$ slots.

\textbf{MISO RVQ codebooks and post-quantization rate}: For each bit allocation $b$, we generate codebook ${\mathcal C}(b)$ by choosing two points uniformly at random from the sphere ${\mathbb C}^2$. For such a codebook ${\mathcal C}(b)$, we compute the ergodic rate over $1000$ standard (zero mean, unit variance), complex Gaussian channel realizations. We repeat this experiment over $100$ codebooks and choose the codebook ${\mathcal C}^*(b)$ that provides maximum ergodic rate. We repeat this procedure for each $b\in\{0,1,\ldots,B\}$ and create a \emph{super-codebook} $\{{\mathcal C}^*(0),\ldots,{\mathcal C}^*(B)\}$. Note that the codebook generation procedure is done once at the beginning of the communication epoch. In the previous section, we proposed
\be
r_{k}[\alpha,b]  = \beta_2(\overline{\text{SNR}})(1-2^{-b})+ \beta_1(\overline{\text{SNR}})2^{-b}\label{ergodicapprox}
\ee
as an approximation for the ergodic rate given $b$ bits. In Fig.~\ref{fig:ergodicrateapprox}, we compare (\ref{ergodicapprox}) with the true (numerically computed) ergodic rate given $b$ bits at various $\overline{\text{SNR}}$ values in a typical operational range. We see that (\ref{ergodicapprox}) is indeed an accurate approximation.
\begin{figure}[hbtp]
\centering
\subfigure[]{
\includegraphics[scale=0.5]{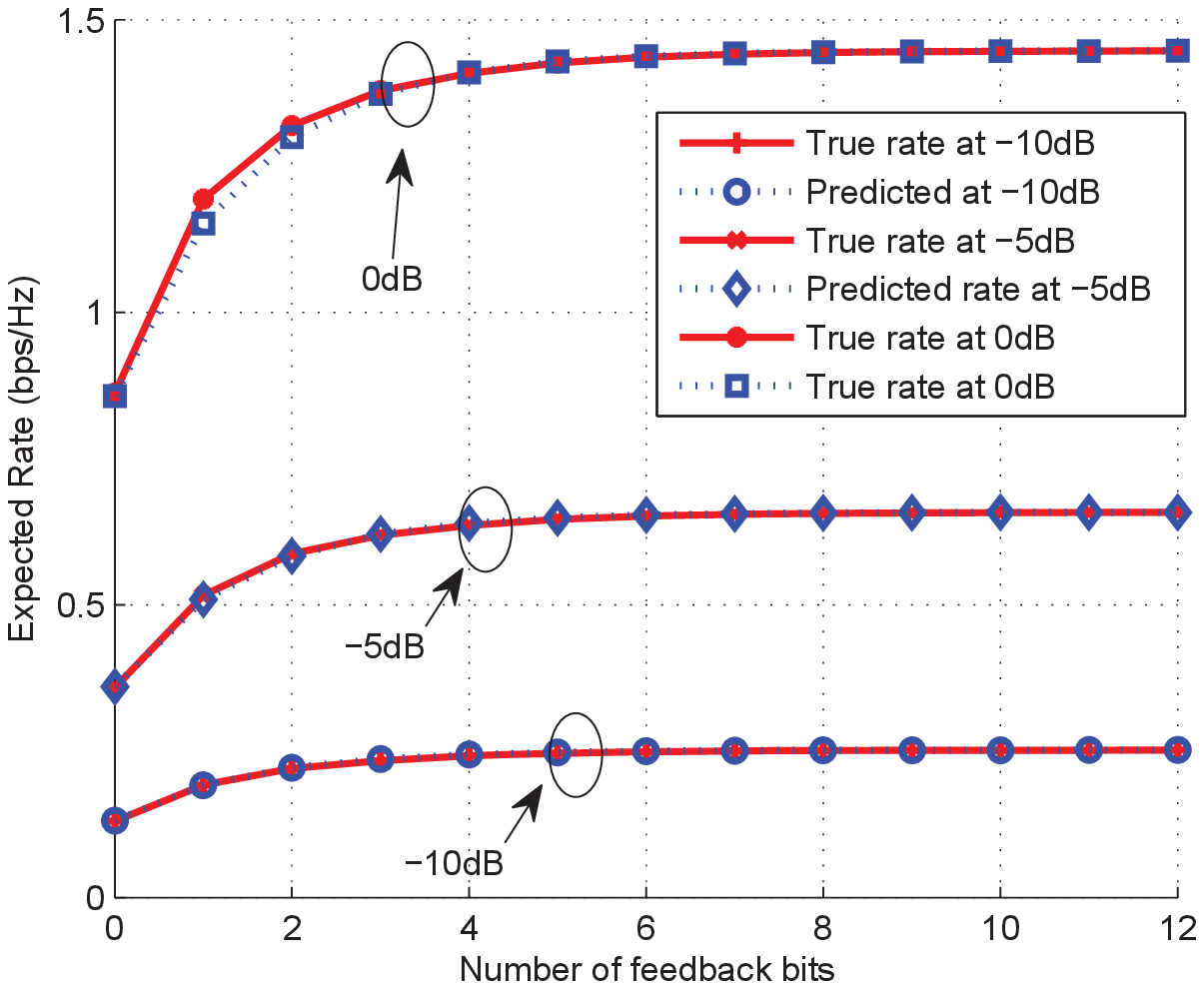}
\label{fig:distortionlowSNR}
}
\hspace{0.1cm}
\subfigure[]{
\includegraphics[scale=0.5]{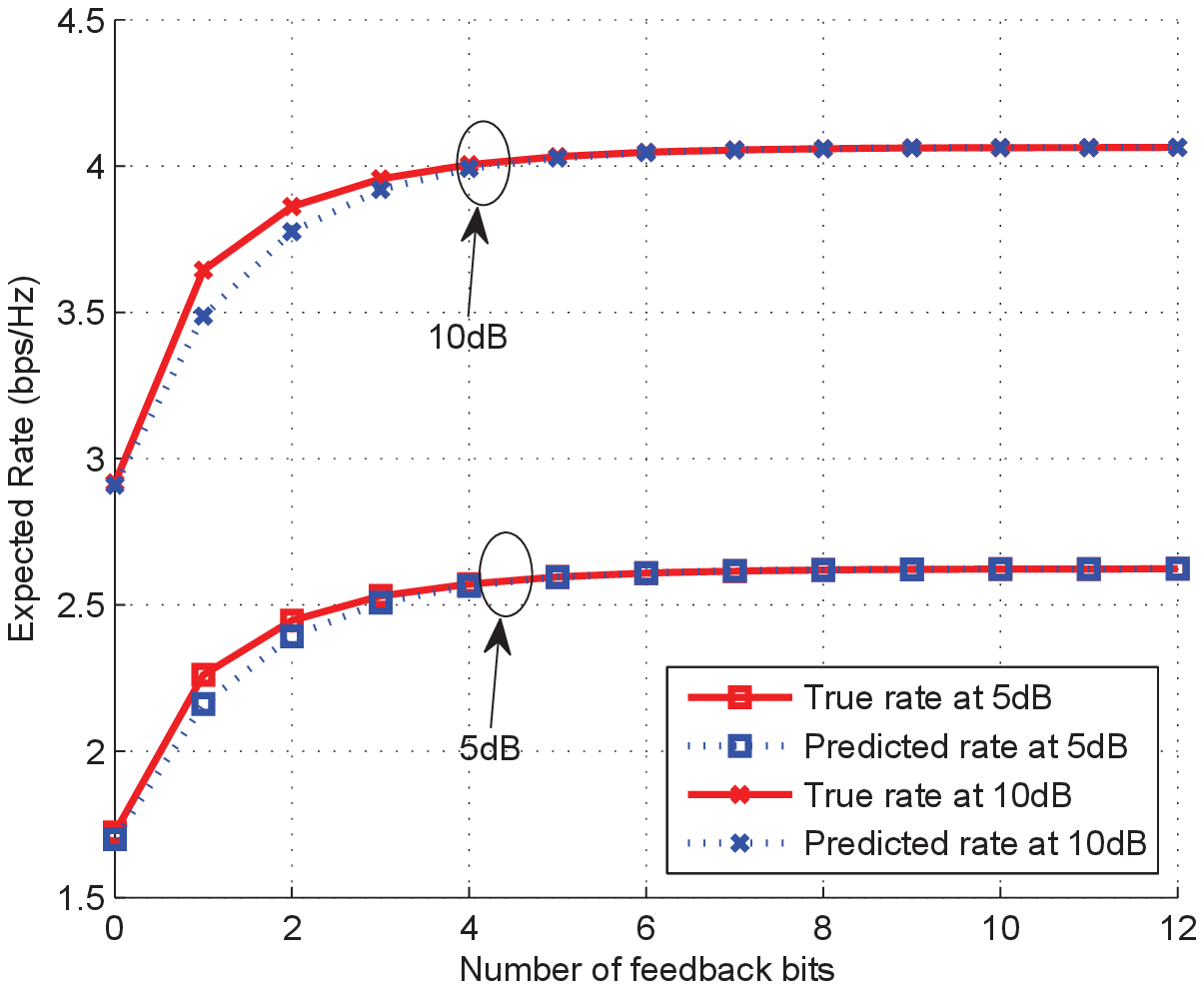}
\label{fig:distortionhighSNR}
}
\vspace{-0.2cm}
\caption{Comparison of predicted rate (\ref{ergodicapprox}) with true numerically computed ergodic rate given for codebooks $\{{\mathcal C}^*(b)\}_{b=0}^B$ at various values of $\overline{\text{SNR}}$; (a) Low-to-moderate $\overline{\text{SNR}}$ (b) Moderate-to-high $\overline{\text{SNR}}$.}
\label{fig:ergodicrateapprox}
\vspace{-0.3cm}
\end{figure}

Having described the simulation setup in detail, we now present the results of our experiments. \revbegin We compare the performance of three algorithms -- the greedy dynamic feedback allocation algorithm in Algorithm~\ref{alg:2}, the equal allocation case, and the case with perfect feedback (i.e., where the bit budget $B = \infty$) -- under the two average SNR profiles. The results for SNR profile with large asymmetry are plotted in Fig.~\ref{fig:throughput}.
\begin{figure}[hbtp]
\centering
\subfigure[]{
\includegraphics[scale=0.5]{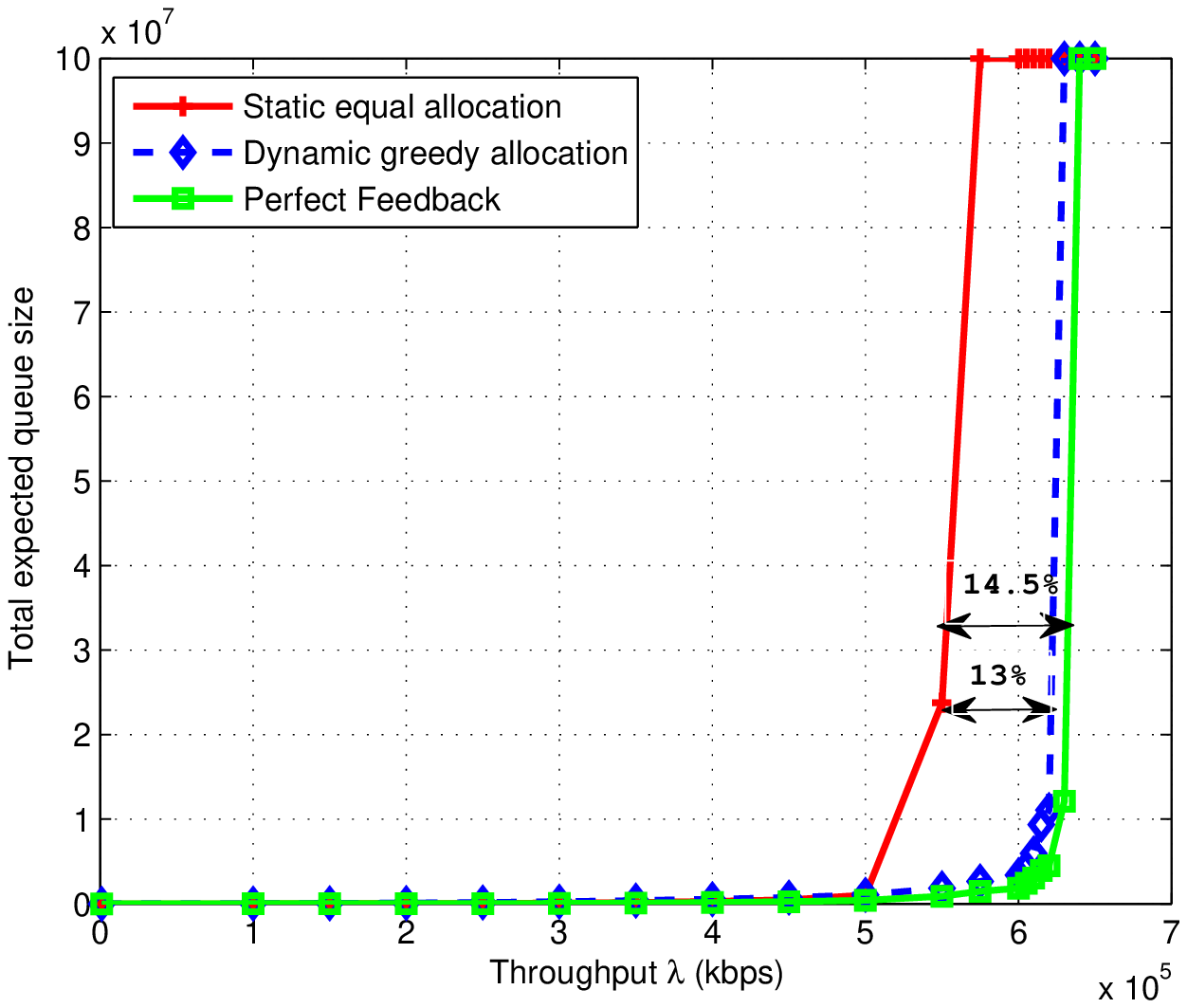}
\label{fig:HighasymmetricResults}
}
\subfigure[]{
\includegraphics[scale=0.5]{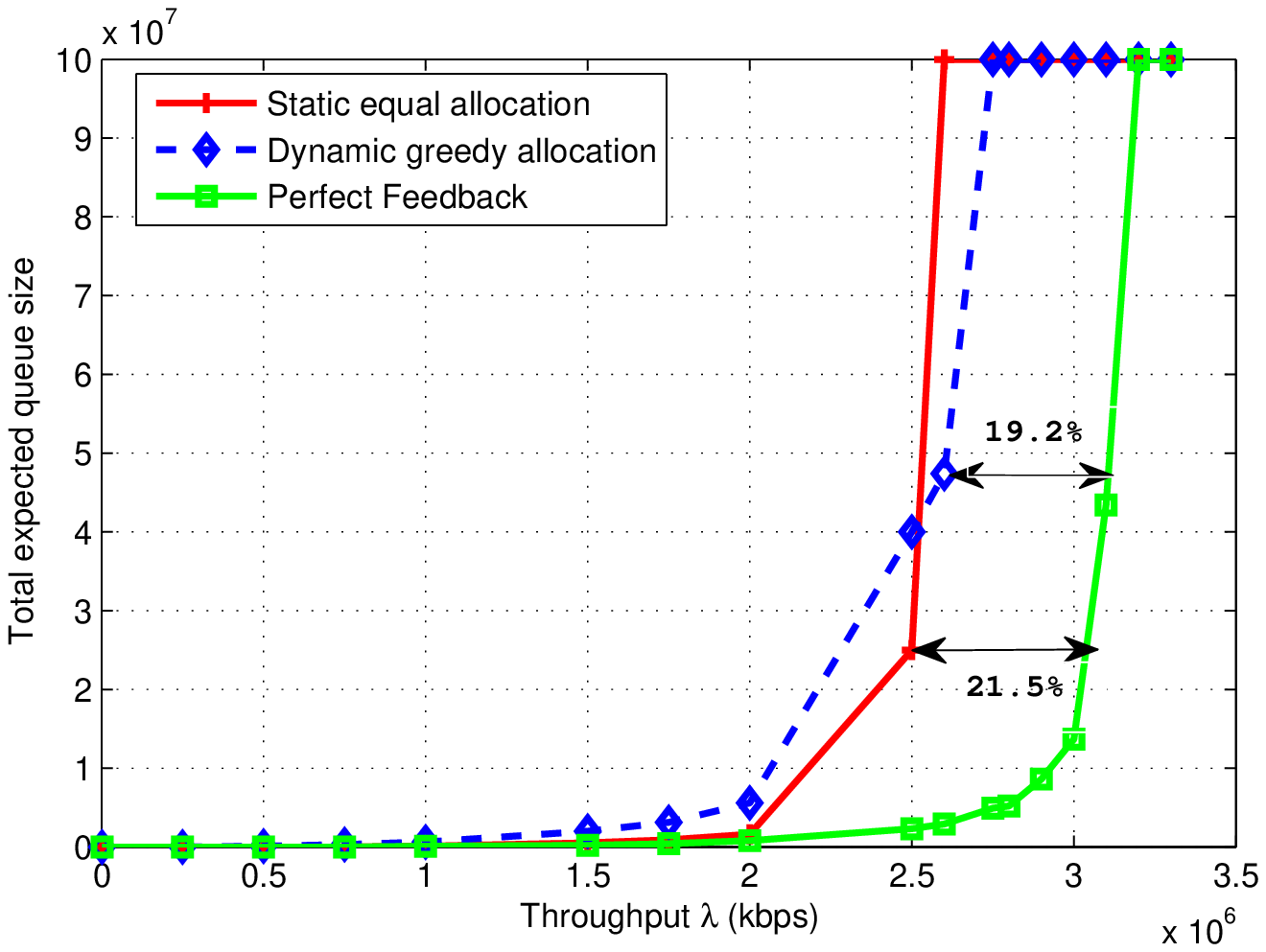}
\label{fig:LowasymmetricResults}
}
\vspace{-0.2cm}
\caption{Throughput under the two feedback schemes with different average SNR profiles. The average queue length is measured over $10000$ iterations; (a) Large asymmetry with $\overline{\text{SNR}}$ profile $\{-10,-8,10,10\}$dB (b) Nearly symmetric case with $\overline{\text{SNR}}$ profile $\{-1,-1,1,1\}$dB}
\label{fig:throughput}
\vspace{-0.3cm}
\end{figure}
In Fig.~\ref{fig:HighasymmetricResults}, we see that the greedy dynamic allocation outperforms the static equal allocation approach by almost $13\%$ while consuming only an additional $\frac{\log_2 \left({B+L-1 \choose L-1}\right) }{T} = 0.88$ bits per second of overhead. Furthermore, greedy dynamic algorithm achieves within $1.5\%$ of the optimal allocation\footnote{The optimal weighted sum-rate is at most as large as the weighted sum-rate with perfect feedback.} through dynamic programming thereby rendering the performance guarantee of $(1-\frac{1}{e})$ in Theorem~\ref{greedycomp} quite conservative.

In the nearly symmetric case however, the gains due to dynamic allocation decrease and almost vanish in the particular experiment that we consider in Fig.\ref{fig:LowasymmetricResults}, as would be expected. We see that in this case, the greedy algorithm achieves within $20\%$ of the optimal.\revend

Thus, with minimal expenditure in overhead, the dynamic allocation approach achieves notable gains in throughput for asymmetric settings, thereby showing considerable promise for systems with larger feedback budgets and a greater degree of asymmetry (in traffic loads and channels).

\section{Concluding remarks}\label{sec:conclusion}

We summarize the algorithmic contributions presented in Sections~\ref{sec:optsol} and~\ref{sec:redcomp} in Table~\ref{table1}. We observe from the table that these algorithms explore the tradeoffs between accuracy, computational efficiency and the structure of the weighted sum-rate function. An interesting question and future direction pertaining to the section on MISO systems is whether such an analysis can be extended to cover other commonly-deployed multiple antenna architectures. Finally, the design of joint data scheduling and feedback allocation policies is another direction for future research.
\begin{table}[h!b!p!]
\setlength{\tabcolsep}{2pt}
\caption{Properties of proposed online feedback allocation algorithms}
\begin{center}
\begin{footnotesize}
\begin{tabular}{ |c |c | c | c |}
\hline			
  Algorithm & Required structure on & Complexity & Approx.  \\
  & weighted sum-rate & & factor \\
  \hline
  Dynamic  & None & ${\mathcal O}(LB^2)$ & 1 \\
  Programming &&&\\
\hline
  Greedy & Non-decreasing  & ${\mathcal O}((B+L)\mbox{log}_2 K)$ & $\left(1-\frac{1}{e}\right)$ \\
   & Sub-modular  & & \\
  \hline
  Convex & Non-decreasing & ${\mathcal O}(L\mbox{log}_2 L)$ & $\frac{1}{2}$ \\
  Relaxation  & Sub-modular   &&\\
    & MISO RVQ Systems   &&\\
  \hline
\end{tabular}
\end{footnotesize}
\end{center}
\label{table1}
\end{table}

In summary, we propose optimal feedback allocation policies for cellular downlink systems where the base station has a limited feedback budget. This problem is solved using dynamic programming incurring pseudo-polynomial complexity in the number of users and the total bit budget. When the weighted sum-rate is a non-decreasing sub-modular function, we leverage the theory of sub-modular function maximization to propose a greedy algorithm with polynomial complexity that has a approximation guarantee of $\left(1-\frac{1}{e}\right)$. For MISO transmit beamforming physical layer communication schemes with quantized beamformer feedback, we recognize that the weighted sum-rate function is non-decreasing and sub-modular for RVQ codebooks. More importantly, it takes a special form that allows us to develop an approximation algorithm based on convex relaxations that can be solved in closed-form, incurring lesser complexity than the greedy algorithm. We connect the performance of the proposed approximate online algorithms to the long-term throughput region of the system.

\appendices
\section{Proof of Theorems~\ref{thm_DPrelaxation}-\ref{thm_DPcomplexity}}\label{proofs_MIMO}
\noi {\em Proof of Theorem~\ref{thm_DPrelaxation}}: The objective function is clearly convex since $2^{-b_k}$ is convex and since linear sums preserve convexity. By studying (\ref{waterfilling_objective}) closely, we can also say that $b_k^*$ is such that $\sum_{k=1}^Lb_k^*=B$ since if this not true, we can increase the bit allocation for at least one user thereby decreasing the objective function. Since $B>0$, $b_k=0,\forall k$ is in the interior of our constraint set ${\mathcal B}$, which implies that Slater's constraint for strong duality is satisfied and that the Karush-Kuhn-Tucker (KKT) conditions are sufficient in nature. The Lagrangian cost function can be written as ${\mathcal L}(b_k,\lambda_k,\eta)=-\sum_{k=1}^L q_k r_1(\text{SNR}_k)2^{-b_k}-\lambda_k b_k+\eta\left(\sum_{k}b_k-B\right)$ for which the KKT conditions are $b_k^*\geq 0$, $\lambda_k^*\geq 0$, $b_k^* \lambda_k=0$, and $\eta^*=q_k r_1(\text{SNR}_k)(\log\bsp 2)2^{-b_k} +\lambda_k^*$. Since $2^{-b}$ is a decreasing function in $b$, it follows that if $\eta^*\leq q_k r_1(\text{SNR}_k)(\log\bsp 2)$, then $\lambda_k^*=0$ and $b_k^*= -\log_2\left(\frac{\eta^*}{q_k (\log\bsp 2)}\frac{1}{r_1(\text{SNR}_k)}\right)$ is a valid solution to (\ref{waterfilling_objective}). If $\eta^*> q_k r_1(\text{SNR}_k)(\log\bsp 2)$, $\lambda_k^*=\eta^*-q_k r_1(\text{SNR}_k)(\log\bsp 2)$ and $b_k^*=0$. Hence, we can write the solution as $b_k^*=\left[-\log_2\left(\frac{\eta^*}{q_k (\log\bsp 2)}\frac{1}{r_1(\text{SNR}_k)}\right)\right]^+$ where $\eta^*$ is chosen such that $\sum_{k}b_k^*=B$. \\

\noi {\em Proof of Theorem~\ref{thm_DPcomplexity}}: In order to compute the solution in Theorem~\ref{thm_DPrelaxation}, we first need to sort $\left\{\theta_k\right\}_{k=1}^L$ in ascending order where $\theta_k = q_k r_1(\text{SNR}_k)(\log\bsp 2)$. This has complexity ${\mathcal O}(L\mbox{log}_2 L)$. Call this sorted set $\left\{\theta_m\right\}$. Once sorted, we need to set $\eta^*=\theta_m$ for each $m$ and test feasibility. Testing feasibility incurs ${\mathcal O}(L)$, as it is a $L$-term addition and scanning through each $\theta_m$ incurs ${\mathcal O}(\mbox{log}_2 L)$ through the use of binary search. As we increase $\eta^*$, more $b_m^*$ terms are set to zero. Once we locate $m_1$ and $m_2$ such that $\eta^{*}=\theta_{m_1}$ is infeasible while $\eta^{*}=\theta_{m_2}$ is feasible, we can compute $\eta^*$ in closed-form since it satisfies $\sum_{m\geq m_2}b_m^*=B$. Hence, the total complexity is ${\mathcal O}(L\mbox{log}_2 L)+{\mathcal O}(L\mbox{log}_2 L)={\mathcal O}(L\mbox{log}_2 L)$.\\


\begin{thebibliography}{99}

\begin{small}

\bibitem{ltebook1}
H. Holma and A. Toskala, ``LTE for UMTS : OFDMA and SC-FDMA based radio access'', {\em Chichester : John Wiley and Sons, Ltd.}, 2009.

\bibitem{intel}
S. Shrivastava and R.  Vannithamby, ``Performance analysis of persistent scheduling for VoIP in WiMAX networks'', {\em Proc. of IEEE Wireless and Microwave Technology Conference (WAMICON) 2009}, pp. 1 - 5, Apr. 2009, Clearwater, FL.

\bibitem{nokia}
H. Wang and D. Jiang, ``Performance comparison of control-less scheduling policies for VoIP in LTE UL'', {\em Proc. Wireless Communications and Networking Conference (WCNC) 2008}, pp. 2497 - 2501, Las Vegas, NV, Apr. 2008.

\bibitem{jin}
H. Jin, C. Cho, N-O. Song and D. K. Sung, ``Optimal rate selection for persistent scheduling with HARQ in Time-Correlated Nakagami-m Fading Channels'', {\em IEEE Trans. Wireless Commun.}, pp. 637 - 647, vol. 10, Feb. 2011.

\bibitem{win}
Li, W.W.-L.  Ying Jun Zhang  So, A.M.-C.  Win, M.Z., ``Slow adaptive OFDMA systems through chance constrained programming'', {\em IEEE Trans. Sig. Proc.}, vol. 58, pp. 3858 - 3869, July 2010.

\bibitem{3gpp}
3GPP, ``3GPP TS 36.300-870'', 3rd Generation Partnership Project, \url{http://www.3gpp.org/ftp/Specs/html-info/36300.htm}.

\bibitem{love_tut}
D. J. Love, R. W. Heath Jr., V. K. N. Lau, D. Gesbert, B. D. Rao and M. Andrews,
``An overview of limited feedback in wireless communication systems'',
{\em IEEE Journ. Sel. Areas Commun.},
vol. 26, pp. 1341-1365, Oct. 2008.

\bibitem{heath}
R. W. Heath, Jr., T. Wu,  and  A. C. K. Soong, ``Progressive refinement for high resolution limited feedback beamforming'', {\em EURASIP Journal on Adv. in Sig. Proc., special issue on Multiuser Limited Feedback}, vol. 2009, Article ID 463823, 2009.

\bibitem{jafarkhani1}
L. Liu and H. Jafarkhani, ``Novel transmit beamforming schemes for time-selective fading multiantenna systems'', {\em IEEE Trans. Sig. Proc.}, vol. 54, pp. 4767-4781, 2006.

\bibitem{jafarkhani2}
L. Liu and H. Jafarkhani, ``Successive transmit beamforming algorithms for multiple-antenna OFDM systems'', {\em IEEE Trans. Wireless Commun.}, vol. 6, pp. 1512–1522, 2007.


\bibitem{mehta}
S. Donthi and N. Mehta, ``Joint performance analysis of channel quality indicator feedback schemes and frequency-domain scheduling for LTE'', {\em IEEE Trans. Vehicular Tech.}, vol. PP, pp. 1-13, 2011.

\bibitem{poor}
M. -O. Pun, J. K. Kyeong and H. V. Poor, ``Opportunistic scheduling and beamforming for MIMO-OFDMA downlink systems with reduced feedback'', {\em IEEE Intern. Conf. on Communications (ICC)}, Bejing, China, May 2008.

\bibitem{gesbert}
D. Gesbert and M.-S. Alouini, ``How much feedback is multiuser diversity really worth?'', {\em in Proc. of Intern. Conf. on Communications (ICC)}, pp. 234-238, Jul. 2004.

\bibitem{chen}
J. Chen, R. A. Berry, and M. L. Honig, ``Limited feedback schemes for
downlink OFDMA based on sub-channel groups'', {\em IEEE Journ. Sel. Areas Commun.}, vol. 26, pp. 1451-1461, Oct. 2008.

\bibitem{agarwal}
R. Agarwal, V. Majjigi, Z. Han, R. Vannithamby and J. Cioffi, ``Low complexity resource allocation with opportunistic feedback over downlink OFDMA networks'',
{\em IEEE Journ. Sel. Areas Commun.}, vol. 26, pp. 1462-1472, Oct. 2008.

\bibitem{kleinberg}
J. Kleinberg and E. Tardos, ``Algorithm Design'', {\em Addison Wesley}, 2006.

\bibitem{love_RVQ1}
A. D. Dabbagh and D. J. Love, ``Feedback rate-capacity loss tradeoff for limited feedback MIMO Systems'', {\em IEEE Trans. Inform. Theory}, vol. 52, pp. 2190-2202, May 2006.

\bibitem{love_grass}
D. J. Love, R. W. Heath Jr. and T. Strohmer,
``Grassmannian beamforming for multiple-input-multiple-output wireless systems'',
{\em IEEE Trans. Info. Theory},
vol. 49, pp. 2735-2747, Oct. 2003.


\bibitem{goundan}
P. R. Goundan and A. S. Schulz, ``Revisiting the greedy approach to submodular set function maximization'', Jan. 2009.

\bibitem{calinescu}
G. Calinescu, C. Chekuri, M. Pal, J. Vondrak,
``Maximizing a submodular set function subject to a matroid constraint (Extended Abstract)'', {\em Lecture Notes In Computer Science, Proc. 12th Intern. Conf. Integer Prog. and Comb. Optimization}, vol. 4513, pp. 182 - 196, Ithaca, NY, 2007

\bibitem{vondrak}
J. Vondr\`ak, ``Submodularity in combinatorial optimization'', {\em Ph.D. thesis}, Charles University, Prague, 2007.

\bibitem{nemhauser}
G. L. Nemhauser and L. A. Wolsey, ``Best algorithms for approximating the maximum of a submodular set function'', {\em INFORMS}, vol. 3, pp. 177-188, Aug. 1978.

\bibitem{ying}
M. Ouyang and L. Ying, ``On scheduling in multi-channel wireless downlink networks with limited feedback'', {\em Proceedings of the Allerton Conference on Communication, Control, and Computing}, Monticello, IN, Oct. 2009.

\bibitem{shroff}
C. Joo, X. Lin, and N. B. Shroff, ``Greedy maximal matching: Performance limits for arbitrary network graphs under the node-exclusive interference model'', {\em IEEE Trans. Auto. Control}, vol. 54, no. 12, pp. 2734--2744, Dec. 2009.

\bibitem{srikant}
A. Gupta, X. Lin and R. Srikant, ``Low-complexity distributed scheduling algorithms for wireless networks'', {\em IEEE/ACM Trans. Networking}, vol. 17, pp. 1846-1859, Dec. 2009.

\bibitem{huang}
J. Huang, V. G. Subramanian, R. Agarwal, and R. A. Berry, ``Downlink scheduling and resource allocation for OFDM systems'', {\em IEEE Trans. Wireless Commun.}, vol. 8, pp. 288-296, Jan. 2009.

\revbegin
\bibitem{tan1}
C. W. Tan, M. Chiang and R. Srikant, ``Fast algorithms and performance bounds for sum rate maximization in wireless networks'', {\em IEEE Infocom 2009}, Rio De Janeiro, Brazil, Apr. 2009.

\bibitem{tan2}
C. W. Tan, M. Chiang and R. Srikant, ``Maximizing sum rate and minimizing MSE on multiuser downlink: Optimality, fast algorithms, and equivalence via max-min SINR'', {\em IEEE Trans. Sig. Proc.}, vol. 59, No. 12, pp. 6127-6143, Dec. 2011.
\revend

\bibitem{tasseph}
L. Tassiulas and A. Ephremides, ``Stability properties of constrained queueing systems and scheduling policies for maximum throughput in multihop radio networks'', {\em IEEE Trans. Auto. Control}, vol. 37, pp. 1936-1949, Dec. 1992.

\bibitem{akrsvw}
M. Andrews, K. Kumaran, K. Ramanan, A. Stolyar, R. Vijayakumar and P. Whiting, ``Scheduling in a queuing system with asynchronously varying service rates'', {\em Probability in the Engineering and Informational Sciences}, vol. 18, pp. 191-217, Apr. 2004.

\bibitem{stolyar_noqueues}
A. L. Stolyar, ``On the asymptotic optimality of the gradient scheduling algorithm for multiuser throughput allocation'', {\em INFORMS}, vol. 53 ,  Issue 1, pp. 12-25, Jan. 2005.

\bibitem{WCDMAbook}
H. Holma and A. Toskala, ``WCDMA for UMTS: Radio access for
third generation mobile communications'', {\em Revised Edition. New York: John Wiley \& Sons}, 2001.

\bibitem{raghavan}
V. Raghavan, M. L. Honig, V. V. Veeravalli, ``Performance analysis of RVQ-based limited feedback beamforming codebooks'', {\em Proc. IEEE ISIT 2009}, pp. 2437-2441, Seoul, Korea, June 2009.

\bibitem{love_RVQ2}
C. K. Au-Yeung and D. J. Love, ``On the performance of random vector quantization limited feedback beamforming in a MISO System'', {\em IEEE Trans. Wireless Commun.}, vol. 6, pp. 458-462, Feb. 2007.

\bibitem{gesbert_MIMO}
D. Gesbert, M. Shafi, D.-S. Shiu, P. J. Smith, and A. Naguib,
``From theory to practice: an overview of MIMO space-time coded wireless systems,''
{\em IEEE Journ. Sel. Areas Commun.}, vol. 21, no. 3, pp. 281-302, 2003.

\end{small}
\end{thebibliography}
\end{document}